\documentclass[runningheads,orivec]{llncs}

\usepackage[T1]{fontenc}
\usepackage{color}
\usepackage[textsize=scriptsize, disable]{todonotes}

\usepackage{graphicx}
\usepackage{subcaption}
\usepackage{tikz} 
\usetikzlibrary{automata,positioning,fit}
\usepackage{xcolor}

\usepackage{amsmath}
\usepackage{amssymb}
\usepackage{stmaryrd} \usepackage{bm}
\usepackage{shuffle} 
\usepackage[capitalise]{cleveref} 

\usepackage{booktabs}
\usepackage{multirow}
\usepackage{amssymb} \usepackage{siunitx} \usepackage{xcolor}

\usepackage[linesnumbered]{algorithm2e}
\RestyleAlgo{ruled}
\SetKwComment{Comment}{/* }{ */} \SetKwInOut{Input}{Input}\SetKwInOut{Output}{Output}
\SetKwFor{For}{for all}{}{}

\newcommand{\pass}{\texttt{Pass}}
\newcommand{\fail}{\texttt{Error}} 
\newcommand{\localerror}[1]{\texttt{LocalError(#1)}} 
\newcommand{\intererror}{\texttt{InterError}} 
\newcommand{\centralerror}{\texttt{CentralError}}

\newcommand{\lang}{\mathcal{L}}
 
\newcommand{\reads}[1]{\xrightarrow{#1}} 
\newcommand{\aut}{\mathcal{A}}

\newcommand{\suff}[1]{\textsf{suf}(#1)}
\newcommand{\area}[1]{\textsf{Area}(#1)}
\newcommand{\central}[1]{\textsf{Central}(#1)}
\newcommand{\decentral}[1]{\textsf{Proj}(#1)}
\newcommand{\ar}[1]{\textsf{area}(#1)} 
\newcommand{\init}{\textsf{init}} 
\newcommand{\interaut}{\textsf{Inter}(\aut, (\mu_i)_i)} 
\newcommand{\inter}{\textsf{Inter}}

\newcommand{\set}[1]{\{ #1 \}}
 
\newcommand{\proj}[1]{{|_{#1}}} 
\newcommand{\size}[1]{{|#1|}}

\newcommand{\noti}{\text{\=\i}}

\begin{document}

\title{Runtime Verification of Interactions Using Automata}
\author{Chana Weil-Kennedy\inst{1}\orcidID{0000-0002-1351-8824} \and
Darine Rammal\inst{1}\orcidID{0000-0001-6863-0656} \and
Christophe Gaston\inst{1}\orcidID{0000-0001-6865-5108} \and
Arnault Lapitre\inst{1}\orcidID{0000-0002-2185-4051}}
\authorrunning{C. Weil-Kennedy et al.}
\institute{Université Paris-Saclay, CEA, List, Palaiseau, France \\
\email{firstname.lastname@cea.fr}}

\maketitle

\begin{abstract}
Runtime verification consists in observing and collecting the execution traces of a system and checking them against a specification, with the objective of raising an error when a trace does not satisfy the specification.
We consider distributed systems consisting of subsystems which communicate by message-passing. Local execution traces 
consisting of send and receive events 
are collected on each subsystem. We do not assume that the subsystems have a shared global clock, which would allow a reordering of the local traces. Instead, we manipulate multitraces, which are collections of local traces. 
We use interaction models as specifications: they describe communication scenarios between multiple components, and thus  specify a desired global behaviour.
We propose two procedures to decide whether a multitrace satisfies an interaction,
based on automata-theoretic techniques.
The first procedure is straightforward, while the second provides more information on the type of error and integrates the idea of reusability: because many multitraces are compared against one interaction, some preprocessing can be done once at the beginning.
We implement both procedures and compare them.

\end{abstract}

\section{Introduction}

Formal verification concerns itself with checking whether a system conforms to a given formal requirement, or \textit{specification}.
Runtime Verification (RV) in particular consists in observing system executions and raising warnings as soon as possible when these executions do not conform to the specification. 
Contrary to model-checking, another classic formal verification technique, RV does not require a model of the system being checked.  

The specifications define acceptable executions and can take various forms.
The authors of \cite{MaheGG20}  proposed using \emph{interaction models} (or simply interactions)  for use as formal specifications in the context of RV. The notion of interaction they use is similar to message sequence charts (MSCs) \cite{MauwR99} or UML Sequence Diagrams (UML-SD) \cite{uml251}. 
Interactions define valid scenarios of message exchanges between multiple components, and as such they are particularly well-suited for specifying communication protocols. 
The interaction language includes high-level operators such as sequencing, parallel composition, choice, and repetition. 
Interactions are especially valuable due to their graphical nature and ease of understanding, making them accessible even to non-experts in formal methods. 
They use graphical conventions familiar to engineers, like vertical lifelines for components, horizontal arrows between them for messages, time flowing from top to bottom and scheduling operators shown through annotated boxes.

In \cite{MaheGG20} and in subsequent articles, an RV approach is given in which the exchange of messages between components, called here \textit{subsystems},
are collected via an observation harness deployed on a central component.
The resulting sequences of message sends and receives are called \textit{traces}, and an algorithm is given to determine whether an observed trace is accepted by a reference interaction, based on an unfolding of the operational semantics of the interaction.

In the context of distributed systems, the existence of a central observation harness can be too strong an assumption. 
A distributed system consists of a collection of subsystems which may be deployed on different hardware resources. 
There may be no global clock allowing a reordering of the sequences of actions of each subsystem into a single trace. To address this, 
the notion of {\em multitrace} is used: a collection of \textit{local} traces, one per subsystem. 
The work in \cite{MaheBGLG21} extends the approach of \cite{MaheGG20} to multitraces.
A multitrace is accepted by an interaction if there exists an interleaving of the local traces which is accepted by the interaction.

In \cite{Formalise2024}, to improve the performance of the trace acceptance algorithm introduced in \cite{MaheGG20}, the authors explored a different direction: precomputing a non-deterministic finite automaton (NFA) from the reference interaction which accepts the same traces. This is feasible for a sublanguage of the interaction language\footnote{The only difference with the full interaction language is that all actions under a \textit{loop} operator must happen before a new iteration of the loop starts.}.
Thanks to this result, checking whether a trace is accepted by an interaction can be reduced to checking its acceptance by the corresponding NFA. However, this approach was not extended to multitraces.

To improve performance one may also take into account \textit{reusability}. 
Past approaches define a procedure for one multitrace (or trace) analyzed against one interaction,
but usually a very large number of multitraces is analyzed against one interaction.
Ideally we want an approach designed given an interaction, that is then efficient at analyzing a large number of multitraces against this interaction.
Another desired feature is the error \textit{information}: 
if an error occurs at a local subsystem level, 
the error message should reflect this.
This is important knowledge for the administrator of the system to direct correction mechanisms.

\paragraph{Contributions}
In this paper, we 
extend \cite{Formalise2024}'s approach to handle multitraces, using automata-theoretic techniques.
\begin{itemize}
    \item \textit{Centralized Procedure:} \ we introduce a straightforward algorithm which decides whether an interleaving of local traces is accepted by an NFA.
    \item \textit{Semi-centralized Procedure:} \ we also propose a  more refined algorithm that starts by performing checks at the subsystem level with the aim of detecting local errors as early as possible. Some errors cannot be found at the subsystem level, as they result from the interleaving of the local traces. To have a complete procedure, we perform a second step in which information generated locally during the first step is used to determine whether the multitrace is accepted.
\end{itemize}
The semi-centralized procedure integrates the idea of
reusability: we first compute local verifiers from
an interaction, and then reuse these for each multitrace
to be analyzed against the interaction, yielding
local checks which take linear time 
in the size of the local trace.
Our centralized and semi-centralized procedures 
are \textit{sound} and \textit{complete} 
with regards to our problem,
meaning that the procedures answer \fail{} 
if and only if there is no interleaving of the
local traces 
that is accepted
by the interaction.
We implement and compare the two procedures, discussing the merits of both.

\paragraph{Organization}
\cref{sec:context} describes our setting and gives the definitions.
\cref{sec:centralized} describes the centralized procedure, and \cref{sec:decentralized} describes the semi-centralized procedure. 
\cref{sec:implementation}
 discusses the implementations
 of the procedures and evaluates them.
\cref{sec:related} discusses related work, 
and 
concludes this paper.

\section{Setting and Definitions}
\label{sec:context}

\subsection{Distributed Systems and Observability}

We consider a system comprised of subsystems
distributed over $n\ge 2$ locations.
Each subsystem sends and receives messages, 
and this describes a set of actions of the form  
$s!m$ (subsystem $s$ sends message $m$)
and $s?m$  (subsystem $s$ receives message $m$).
Several subsystems may be at a same location $i$,
but each subsystem is at exactly one location.
We write $\Sigma_i$ the set of actions for location $i$.
Notice that the $\Sigma_i$ are disjoint for $i\in\set{1,\ldots,n}$.
We write $\Sigma$ the union of the $\Sigma_i$ 
for all locations in the system.
We assume that there are a finite number of subsystems
and a finite number of messages\footnote{These can be seen as message types, without regard for possible data being sent.}.

While the system is running, 
observation harnesses placed at each location
collect execution traces, called \emph{local traces},
which are sequences of actions of $\Sigma_i$.
We assume that there is a \emph{local clock} at each location,
and thus that local traces are correctly ordered
(this can be ensured by time-stamping the actions with the local clock).
We make the assumption that local traces
are correctly collected, that is there is \emph{no loss} or faulty observation.
Finally, we make the assumption that there are message markers,
for example session IDs,
which allow the local observation harnesses to organize  what they collect into \emph{finite} local traces, one per message group.
We manipulate \emph{multitraces} made up of $n$ finite local traces,
one per location, all corresponding to the the same group. 
This last assumption corresponds to common practice in 
distributed communication protocols: a session is a temporary, interactive exchange of information between participants communicating across a network. Protocols that use sessions mark all messages of participants in a same session with the session ID
(see e.g. standard \cite{session-protocol}).

\subsection{Interactions} 

\begin{figure}[t]
    \centering
    \includegraphics[scale=0.3]{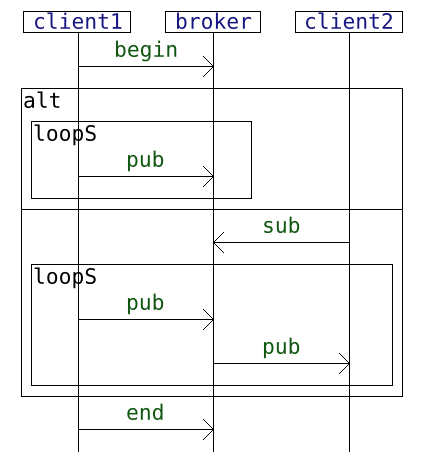} 
    \caption{Graphical representation of an interaction for an MQTT-style protocol.} 
    \label{fig:interaction-example}
\end{figure}

Figure \ref{fig:interaction-example} 
depicts an interaction model with three subsystems $client_1,client_2$ and $broker$, that we suppose in three different locations.
Each subsystem's behaviour is represented by a vertical line called a \emph{lifeline}, 
and messages are horizontal arrows.
The interaction models a simple MQTT-like protocol \cite{mqttv3}:
$client_1$ starts a topic
hosted on the broker's server by sending \textsf{begin};
$client_1$ then publishes regularly to the topic.
If $client_2$ subscribes to the topic,
$broker$ will send it the publications it receives.
The topic ends when $client_1$ sends \textsf{end}
to the broker.

Formally,
let $S$ a finite set of subsystems, 
let $M$ a finite set of messages.
Let $\Sigma(S,M)$ be the finite set of \emph{actions}
$s!m$ and $s?m$ for all $s\in S, m\in M$.

\begin{definition}
Given $S$ and $M$ two finite sets,
an \emph{interaction}
is constructed by the following grammar:
$$I := a \ | \ loop_S(I) \ | \ strict(I_1,I_2) \ | \ alt(I_1,I_2) \ | \ cr_s(I_1,I_2)$$
where $a\in \Sigma(S,M)$ and $s\subseteq S$.
\end{definition}

Intuitively, $strict$ is the strict sequencing 
operator: $strict(i,i')$ enforces that 
actions of $i'$ happen after $i$\footnote{In the graphical representation, a horizontal arrow labeled $m$ from $l_1$ to $l_2$ represents $strict(l_1!m,l_2?m)$.}.
Operator $loop_S(i)$ denotes an arbitrary number 
of repetitions (including zero) of 
$i$ that must finish entirely before
a new repetition starts.
Operator $alt$ denotes exclusive choice
between two (sub)interactions.
Operator $cr_s$, for $s$ a subset of the subsystems,
is called the ``concurrent region'' operator:
$cr_s(i,i')$ indicates that actions of
$i$ and $i'$ on $s$'s lifeline can be interleaved.
Notice that 
there may be a send without a corresponding receive and vice-versa
(like message $m_4$ in \cref{fig:interaction} below).

The semantics of an interaction $I$
is a set of traces $\sigma(I)$;
we say that a trace of $\sigma(I)$ 
is \emph{accepted} by $I$.
The definition of $\sigma(I)$ was originally
given in a structural operational style
using rules that define an expression relation
$\reads{}$ such that 
if $t \in \sigma(i')$ and $i\reads{a}i'$
then $a.t \in \sigma(i)$.
The rules can be found in Section 3.2 
of \cite{Formalise2024}.
However, we will work with a different formalization of the semantics, given as a non-deterministic finite automaton (NFA).
We first define NFAs, then give the link between interactions and NFAs.

\subsection{Words and Automata}

An \emph{alphabet} $\Sigma$ is a finite set of symbols, 
called letters.
A finite \emph{word} is a finite sequence of letters.
The length $|w|$ of a word $w$ is its number of letters.
The set of words of length $n\geq 0$ is $\Sigma^n$;
 $\Sigma^*$ denotes the Kleene closure of $\Sigma$, 
 i.e. the union of all finite words over $\Sigma$.
We denote $\varepsilon$ the empty word of length $0$.

Word $w''\in\Sigma^*$ is a \emph{suffix}
of word $w\in\Sigma^*$ if there exists 
$w'\in\Sigma^*$ such that $w=w'w''$.
The set of 
suffixes of $w$ are denoted $\suff{w}$.
The \emph{projection} of a word $w\in\Sigma^*$
onto $\Gamma \subseteq \Sigma$ 
is obtained by removing all
letters not in $\Gamma$ from $w$, and 
it is denoted $w_\proj{\Gamma}$.
Given a set of words $W\subseteq\Sigma^*$,
$W_\proj{\Gamma}$ is the set of 
$w_\proj{\Gamma}$ for $w\in W$. 
When $\Gamma$ is a location alphabet $\Sigma_i$,
we write $W_\proj{i}$ for $W_\proj{\Sigma_i}$.

\begin{definition}
\label{def:automaton}
A \emph{non-deterministic finite automaton} (NFA)
is $\aut = (Q, \Sigma, \delta, Q_0, \allowbreak F)$, 
where $Q$ is a finite set of states,
$\Sigma$ is a finite alphabet,
$\delta \subseteq Q\times\Sigma\times Q$ 
is a set of transitions denoted
$q_1 \reads{a} q_2$,  
$Q_0\subseteq Q$ are the initial states
and $F\subseteq Q$ are the final states.
We write $\delta(q,a)$ the set of states reachable 
from state $q$ by reading letter $a$.
A \emph{deterministic finite automaton} (DFA) 
is an NFA in which $\delta(q,a)$ is either empty or a singleton
for every $q\in Q, a\in \Sigma$.
\end{definition}

Given an NFA $\aut$,
we let $Q^\aut$ denote the states of $\aut$,
$\delta^\aut$ its transitions,
$Q_0^\aut$ its initial states and
$F^\aut$ its final states.
We extend notation $\delta(q,a)$: 
given $w\in \Sigma^*$, $\delta(q,w)$ is 
the set of states reachable from state $q$ 
by reading $w$.
Given $S\subseteq Q$, $\delta(S,a)$ is 
the union of the $\delta(q,a)$ such that $q\in S$.

A \emph{path} in an NFA $\aut$ is a sequence of 
transitions of $\aut$ which is either empty
or of the form $q_1 \reads{a_1} q_2, q_2 \reads{a_2} q_3, \ldots, q_n \reads{a_n} q_{n+1}$.
We say this path \emph{reads} $\varepsilon$ 
or $a_1 \ldots a_n$ respectively.
A path is \emph{accepting} if it starts 
in an initial state and ends in a final state.
A word $w\in\Sigma^*$ is \emph{accepted} 
by an NFA if there exists 
an accepting path that reads $w$.
The \emph{language} recognized by an NFA,
denoted $\lang(\aut)$, 
is the set of accepted words (possibly empty).
A set of words $\lang \subseteq \Sigma^*$
is called a \emph{regular language} if it is
recognized by an NFA.

\begin{definition}
A sub-automaton of an NFA $\aut$
is an NFA defined by a subset 
of the transitions of $\aut$.
Formally, we call NFA 
$\mathcal{B} = (Q, \Sigma^\aut, \delta, Q_0, F)$ 
a \emph{sub-automaton} of $\aut$ if 
$\delta \subseteq \delta^\aut$,  
$Q$ is the states of $Q^\aut$ appearing in $\delta$,
i.e. $Q= \set{q|(q,a,q')\in \delta^\aut} \cup \set{q'|(q,a,q')\in \delta^\aut}$,
$Q_0=Q_0^\aut \cap Q$
and $F =F^\aut \cap Q$.
\end{definition}

Given two regular languages $\lang_1,\lang_2$,
the \emph{shuffle} of $\lang_1,\lang_2$,
denoted $\lang_1 \shuffle \lang_2$,
contains the words $w$ such that $w$ is
an interleaving of a word $w_1\in\lang_1$ and 
a word $w_2\in\lang_2$.
Formally, given $\lang_1$ over alphabet $\Sigma_1$
and $\lang_2$ over $\Sigma_2$, 
$\lang_1 \shuffle \lang_2$ is the language
$\set{u_1 v_1 u_2 v_2 \ldots u_k v_k \ | \ \forall i\in\set{1,\ldots,k}, u_i \in \Sigma_1^*, v_i \in \Sigma_2^*, w_1=u_1\ldots u_k\in\lang_1, w_2=v_1\ldots v_k\in\lang_2}.$
The shuffle of two regular languages is regular
\cite{GinsburgS65}.
We can generalize the shuffle operator to 
an arbitrary number of languages using
$\lang_1 \shuffle \ldots \shuffle \lang_k \shuffle \lang_{k+1}
= (\lang_1 \shuffle \ldots \shuffle \lang_k) \shuffle \lang_{k+1}$ for $k\ge 2$.

\subsubsection{Link with interactions}
Paper \cite{Formalise2024} shows 
how to synthesize an NFA $\aut$ 
from an interaction $I$
such that the language of the NFA is 
exactly the traces accepted by the interaction.

\begin{theorem}[Theorem 4.8 of \cite{Formalise2024}]
Given an interaction $I$, one can synthesize an NFA $\aut$
such that $\lang(\aut) = \sigma(I)$.
\end{theorem}
The resulting automaton has exactly one initial state.
In the following, we work with interactions
given as NFAs.

\begin{figure}
	\centering
	\begin{subfigure}[t]{.45\textwidth}
	\centering
    \includegraphics[scale=0.3]{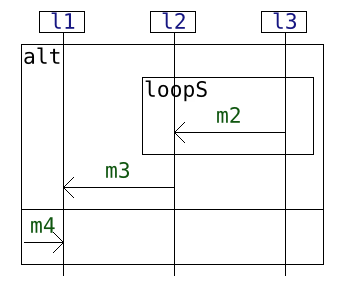} 
    \caption{An interaction.} 
    \label{fig:interaction}
    \end{subfigure}
~
    \begin{subfigure}[t]{.45\textwidth}
	\centering
    \begin{tikzpicture}[node distance=2cm, auto, initial text={}, scale=0.8, transform shape]
		\node[state, initial] (q0) {0};
		\node[state, below left of=q0] (q1) {1};
		\node[state, below=2 of q0] (q2) {2};
		\node[state, accepting, right of=q2] (q3) {3};
		\node[state, left of=q2] (q4) {4};

		\path[->]
		(q0) edge[] node[above left] {$l_3!m_2$} (q1)
		(q0) edge[] node {$l_2!m_3$} (q2)
		(q0) edge[] node {$l_1?m_4$} (q3)
		(q2) edge[] node[below] {$l_1?m_3$} (q3)
		(q4) edge[bend left=20] node[left] {$l_3!m_2$} (q1)
		(q1) edge[bend left=20] node[right] {$l_2?m_2$} (q4)
		(q4) edge[] node[below] {$l_2!m_3$} (q2)
		;
\end{tikzpicture}
    \caption{The corresponding NFA}
    
    \label{fig:nfa-actions}
	\end{subfigure}
	\caption{A graphical representation of an interaction, and its corresponding NFA.}
\end{figure}
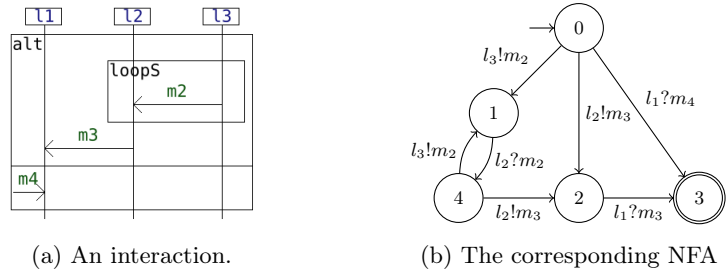
\vspace{-5mm}

\subsection{Verification problem}

The problem we consider is, 
given a multitrace and an interaction,
whether there exists an interleaving of 
the local traces of the multitrace 
that is accepted by the interaction.
This is the same problem as in 
previous work \cite{MaheBGLG21}
(where it is formulated as the equivalent: whether there exists a trace accepted by the interaction that projects onto the multitrace).

\begin{example}
\label{ex:interleaving}
Consider the interaction of \cref{fig:interaction-example},
and the multitrace $(\mu_1,\mu_2,\mu_3)$
with $\mu_1 = clt_1!begin \ clt_1!pub \ clt_1!end$,
$\mu_2 = brok?begin \ brok?pub \ brok?end$
and $\mu_3=\varepsilon$ (the empty trace).
An interleaving that is accepted by the interaction
is $clt_1!begin \ clt_1!pub \ brok?begin \ brok?pub \ clt_1!end \ brok?end$ (we are in the second branch of the $alt$,
and the $loop_S$ is taken once).
Notice that $brok$ receives $begin$ after $clt_1$ has 
already sent both $begin$ and $pub$. 
Swapping $clt_1!pub$ and $brok?begin$ would also yield 
an accepted interleaving.
A non-accepted interleaving example is when $brok$ receives $begin$ before $clt_1$ sends it.
\end{example}

We say that a multitrace $(\mu_i)_{i=1}^n$ 
is \emph{locally correct}
with regards to an interaction $I$
if for each local trace there exists a
 trace accepted by $I$ that projects onto it,
i.e. for all $i$, there exists $w \in \sigma(i)$
such that $w_{|i} = \mu_i$.
A locally correct multitrace 
is not always accepted by $I$.

\begin{example}
Consider the multitrace $(\mu_1,\mu_2,\mu'_3)$
where $\mu'_3=clt_2!sub$ and $\mu_1,\mu_2$ 
are as in \cref{ex:interleaving}.
It is locally correct,
but no accepted interleaving exists
because $\mu_3$ corresponds to the first branch of the $alt$
and $\mu_2$ to the second.
Notice that $\mu_1$ does not allow subsystem $clt_1$
to know which branch it is in.
\end{example}

\paragraph{Verification problem in an automata-theoretic light} For a multitrace $(\mu_i)_{i=1}^n$,
the set of interleavings of the $\mu_i$
is the shuffle language 
$\mu_1 \shuffle \mu_2 \shuffle \ldots \shuffle \mu_n $.
Given an interaction as an automaton $\aut$,
we can thus express our verification problem 
as checking 
$\mu_1 \shuffle \mu_2 \shuffle \ldots \shuffle \mu_n \cap \lang(\aut) = \emptyset.$
If there is no element in the intersection
of the shuffle product with the language of $\aut$,
then no interleaving of the local traces 
is accepted by the interaction.

We informally state what we consider large and small values in our setting. 
We consider that there is a \textit{large}
number of multitraces that is checked against 
one same interaction.
We also consider the length of the multitraces
to be \textit{large}.
The number of locations $n$ and the size of 
the interaction NFA $\aut$ are \textit{small}
compared to the two preceding values.

\section{A Centralized Procedure}
\label{sec:centralized}

We have determined in \cref{sec:context} that,
given our interaction as an automaton $\aut$ and
given a multitrace $(\mu_i)_{i=1}^n$,
our verification problem reduces to checking
$\shuffle_{i=1}^n \ \mu_i \cap \lang(\aut) = \emptyset.$
To this end, we construct an automaton 
recognizing language
$\shuffle_{i=1}^n \ \mu_i \cap \lang(\aut)$, 
stopping if we mark a state as final.

\begin{definition}
Let $\aut= (Q^\aut, \Sigma, \delta^\aut, \set{q_0^\aut}, F^\aut)$ an NFA, 
let $\mu_i\in\Sigma_i^*$ for all $i\in[1,n]$.
Then $\central{\aut, (\mu_i)_i} = (Q,\Sigma,\delta,\set{q_0},F)$ 
is defined as
\begin{itemize}
    \item $Q = \suff{\mu_1} \times \ldots \times \suff{\mu_n} \times Q^\aut$
    \item $\delta$ is such that 
    $(w_1,\ldots,w_n,q) \reads{a} (w_1',\ldots,w_n',q')$
    if and only if $q' \in \delta^\aut(q,a)$ and
    there exists $i$ such that $w_i=a w_i'$ and 
    for all $j\neq i$, $w_j=w_j'$
    \item $q_0 = (\mu_1,\ldots,\mu_n,q_0^\aut)$
and $F = \set{(\varepsilon,\ldots,\varepsilon,q) \ | \ q \in F^\aut}$
\end{itemize}
\end{definition}

The states of $\central{\aut, (\mu_i)_i}$
are tuples made up of a suffix of each $\mu_i$ 
and one state of $\aut$.
An $a$-transition between states corresponds to
taking an $a$-transition in $\aut$ and 
reading an $a$ on one of the suffixes.
The automaton starts in the state in which
all words $\mu_i$ are complete, and 
progressively reads an interleaving of the words $\mu_i$
while advancing in $\aut$.

\begin{theorem}
\label{thm:central}
Let $\aut = (Q^\aut, \Sigma, \delta^\aut, \set{q_0^\aut}, F^\aut)$ an NFA, 
let $\mu_i$ a word over $\Sigma_i$ for all $i\in[1,n]$.
Then $\central{\aut, (\mu_i)_i}$ recognizes
the language $\shuffle_{i=1}^n \ \mu_i \cap \lang(\aut)$.
\end{theorem}

\begin{proof}
The shuffle of $n\ge 2$ languages $\lang_1, \ldots, \lang_n$
is recognized by the automaton 
$(Q,\Sigma,\delta,Q_0,F)$ 
where $Q\subseteq (\Sigma^*)^n$ is the states $(w_1,\ldots,w_n)$
for $w_i$ a suffix of a word in $\lang_i$ for all $i$,
$Q_0$ is the states $(w_1,\ldots,w_n)$
for $w_i$ a word in $\lang_i$ for all $i$, and 
$F$ is the state $(\varepsilon,\ldots,\varepsilon)$
(called the \emph{naive automaton} in \cite{EremondiIM18}).
We adapt this automaton to the case where
the $\lang_i$ are the singletons $\set{\mu_i}$,
and intersect it with the NFA for $\aut$
to obtain $\central{\aut, (\mu_i)_i}$.
Notice that in the singleton language shuffle
there is a unique initial state 
$(\mu_1,\ldots,\mu_n)$. 
We take $\aut$ with a unique initial state
(since it comes from an interaction),
thus ensuring a single initial state in $\central{\aut, (\mu_i)_i}$.
\end{proof}

\subsection{Procedure}
\label{sec:proc-central}

To check whether there exists 
an interleaving of the $(\mu_i)_{i=1}^n$
accepted by $\aut$,
we build $\central{\aut, (\mu_i)_i}$
starting in the initial state and
building states reachable in one transition
from states already built.
If we build a final state, then we stop
and return \pass.
Otherwise, we complete the construction and return \fail.
This is illustrated in Algorithm \ref{alg:central}.

\vspace{-5mm}
\begin{algorithm}[hbt]
\caption{Centralized}
\label{alg:central}
\Input{NFA $\aut = (Q^\aut, \Sigma, \delta^\aut, \set{q_0^\aut}, F^\aut)$, multitrace $(\mu_1,\ldots,\mu_n)$}
\Output{\pass{} or \fail{}}
$q_0 \gets (\mu_1,\ldots,\mu_n,q_0^\aut)$\;
$F \gets \set{(\varepsilon,\ldots,\varepsilon,q) \ | \ q \in F^\aut}$\;
$W \gets \set{q_0}$\;
\While{$W \neq \emptyset$}{
  \textbf{pick} $(w_1,\ldots,w_n,q)\in W$\;
    \For{$(w_1',\ldots,w_n',q')$ \textbf{such that} $q' \in \delta^\aut(q,a)$ \textbf{and} $\exists i$, $w_i=a w_i'$ \textbf{and} $\forall j\neq i$, $w_j=w_j'$}{
    \uIf{$(w_1',\ldots,w_n',q') \in F$}{
        \KwRet{} \pass{}
        }
  \textbf{else add} $(w_1',\ldots,w_n',q')$ \textbf{to} $W$\; 
    }
}
\KwRet{} \fail{}
\end{algorithm}
\vspace{-5mm}

\paragraph{Soundness and Completeness}
Our procedure answers \pass{} if and only if
there exists a state reachable from $q_0$ 
in $\central{\aut, (\mu_i)_i}$ which is final
i.e. if and only if there is an accepting word.
By \cref{thm:central}, this ensures that 
the procedure is sound and complete: 
it answers \pass{} if and only if
there exists an interleaving of the $\mu_i$
that is recognized by $\aut$.

\paragraph{Complexity}
Our procedure is deterministic and 
has quadratic complexity.
Given $\aut$
and words $(\mu_i)_{i=1}^n$ over alphabet $\Sigma$, 
$\central{\aut,(\mu_i)_i}$ has  
at most $\prod_i (|\mu_i|+1) \times |Q^\aut|$ states
(each $\mu_i$ has $(|\mu_i|+1)$ suffixes).
Thus there are at most 
$|\Sigma| (\prod_i (|\mu_i|+1) \times |Q^\aut|)^2$
transitions.
The procedure builds this automaton,
checking for emptiness at the same time.

\paragraph{Evaluation}
In case of an error (a multitrace of 
which no interleaving is accepted by $\aut$),
the procedure must build the full automaton
$\allowbreak \central{\aut,(\mu_i)_i}$.
Additionally, recall that in our setting 
we compare a large number of multitraces
against (the NFA of) one interaction.
This procedure starts "from scratch" for
every multitrace, it is not reusable.
Finally, it does not give information
on whether the error is local or not.
We propose an approach that addresses 
these issues in the next section.

\section{A Semi-Centralized Procedure}
\label{sec:decentralized}
For this section, 
we fix an NFA
$\aut$ representing an interaction.
The alphabet of $\aut$ is 
$\Sigma = \Sigma_1 \uplus \ldots \uplus \Sigma_n$,
where $\Sigma_i$ is the alphabet of location $i$.

Our semi-centralized procedure is in two phases.
The first phase is local: given a local trace $\mu_i$,
local verifier at location $i$ reads $\mu_i$
on a pre-computed automaton $\textsf{Proj}(i)$
which recognizes the words accepted by $\aut$ projected onto $\Sigma_i$.
If the trace is not accepted, 
the local verifier sends \localerror{i} 
to the central verifier and the procedure stops.
The transitions of $\textsf{Proj}(i)$ are labeled
so as to produce an output automaton $\area{\mu_i}$
upon reading a word $\mu_i$.
This $\area{\mu_i}$ 
accepts words of $\aut$ projecting onto $\mu_i$.
The second phase (which happens if no \localerror{i}  was sent)
consists in centralizing the $\area{\mu_i}$ for each $i$
and intersecting them to produce automaton $\textsf{Inter}$.
$\textsf{Inter}$ is a subautomaton of $\aut$ 
that contains all the paths corresponding to the $\mu_i$.
The last step is to apply the centralized procedure
to $\textsf{Inter}$.

The local verifiers are set-up by constructing
the labeled NFAs $\textsf{Proj}(i)$,
independent of any trace.
This is explained in \cref{sec:local-verifier}, and
then \cref{sec:local-check} and \cref{sec:central-check}
give the definitions and theorems for the two steps of the procedure.
Finally, \cref{sec:decentral-correct} recapitulates the procedure, shows it is sound and complete,
and discusses its complexity.

\subsection{Local Verifiers}
\label{sec:local-verifier}

We construct a deterministic finite automaton 
which recognizes the language $\lang(\aut)_{|_i}$ 
based on a classic algorithm for conversion from 
$\varepsilon$-NFA to DFA.
Additionally, we label the transitions of this
DFA such that each transition 
is labeled with the \textit{area} of $\aut$ 
that it corresponds to.

Let $\Sigma_\noti$ be the letters of $\Sigma$ 
that are not in $\Sigma_i$, i.e. 
$\Sigma_\noti=\Sigma \setminus \Sigma_i$ .
Given an NFA $\aut = (Q, \Sigma, \delta, q_0, F)$,
we define $\delta_\noti(q)$
the set of states that can be reached from $q \in Q$
by reading a (possibly empty) sequence of transitions over letters of $\Sigma_\noti$, i.e.
$\delta_\noti(q) = \set{q' \ | \ q \reads{w} q' \text{ and } w\in\Sigma_\noti^*}.$
For $S$ a set of states,
$\delta_\noti(S)$ denotes the union of
the $\delta_\noti(q)$ such that $q\in S$.

An $\varepsilon$-NFA is an NFA 
whose alphabet contains the empty word $\varepsilon$,
meaning it has transitions labeled $\varepsilon$.
We mimic the classic algorithm for converting
an $\varepsilon$-NFA to a DFA 
(see e.g. \cite[Section 2.5.5]{HopcroftMotwaniUllman})
except that instead of taking an $\varepsilon$-NFA
as input, we take our NFA $\aut$
and treat transitions labeled by letters of 
$\Sigma_\noti$ like $\varepsilon$-transitions.
The output DFA $\decentral{\aut,\Sigma_i}$ 
is defined below.
Intuitively, taking an $a$-transition in $\decentral{\aut,\Sigma_i}$
corresponds to taking a sequence of transitions
$\reads{a} \reads{w}$ with $w\in\Sigma_\noti^*$
in $\aut$.

\begin{definition}
\label{def:decentral-projection}
Let $\aut = (Q , \Sigma, \delta , \set{q_0} , F )$ 
an NFA over $\Sigma=\Sigma_i\uplus\Sigma_\noti$.
Then $\allowbreak \decentral{\aut,\Sigma_i} = (Q^P,\Sigma_i,\delta^P,\set{q_0^P},F^P)$ 
is the DFA defined as
\begin{itemize}
    \item $Q^P = 2^Q$, i.e. the subsets of $Q$
    \item $\delta^P$ is such that 
    $S \reads{a} S'$ if and only if 
    $S' = \delta_\noti (\delta (S,a))$, 
    for $a \in \Sigma_i$
    \item $q_0^P = \delta_\noti (q_0 )$, i.e. the subset of states reachable from $q_0$ in $\aut$ by reading a sequence of letters in $\Sigma_\noti$
    \item $F^P = \set{S \ | \ F  \cap S \neq \emptyset}$, i.e. the subsets of $Q$ containing a final state of $\aut$.
\end{itemize}
\end{definition}

\begin{figure}[t]
	\centering
	\begin{subfigure}[t]{.45\textwidth}
	\centering
    \begin{tikzpicture}[node distance=2cm, auto, initial text={}, scale=0.7, transform shape]
		\node[state, initial] (q0) {0};
		\node[state, below left of=q0] (q1) {1};
		\node[state, below=2 of q0] (q2) {2};
		\node[state, accepting, right of=q2] (q3) {3};
		\node[state, left of=q2] (q4) {4};

		\path[->]
		(q0) edge[] node {$b$} (q1)
		(q0) edge[] node {$c$} (q2)
		(q0) edge[] node {$a$} (q3)
		(q2) edge[] node {$e$} (q3)
		(q4) edge[bend left=20] node {\(b\)} (q1)
		(q1) edge[bend left=20] node {\(d\)} (q4)
		(q4) edge[] node {$c$} (q2)
		;
\end{tikzpicture}
    \caption{NFA $\aut$} 
    \label{fig:nfa}
    \end{subfigure}
~
    \begin{subfigure}[t]{.45\textwidth}
	\centering
    \begin{tikzpicture}[node distance=2cm, auto, initial text={}, scale=0.7, transform shape]
		\node[state, initial, accepting] (q0) {0,1,3};
		\node[state, accepting, below left of=q0] (q1) {2,3};
		\node[state, below right of=q0] (q2) {1,4};

		\path[->]
		(q0) edge[] node {$c$} (q1)
		(q0) edge[] node {$d$} (q2)
		(q2) edge[] node {$c$} (q1)
		(q2) edge[loop right] node {$d$} (q2)
		;
\end{tikzpicture}
    \caption{The DFA $\decentral{\aut,\set{c,d}}$}
    \label{fig:proj-nfa}
	\end{subfigure}
	\caption{Illustration of \cref{def:decentral-projection}.}
	
	\label{fig:nfa-and-proj}
\end{figure}

\begin{example}
\label{ex:proj}
Figure \ref{fig:nfa}
represents an NFA $\aut$, which is the NFA of 
Figure \ref{fig:nfa-actions} with letters for readability.
The alphabets for subsystems $l_1,l_2,l_3$
are respectively $\Sigma_1=\set{a,e},\Sigma_2=\set{c,d}$ and $\Sigma_3=\set{b}$.
Figure \ref{fig:proj-nfa} shows DFA $\decentral{\aut,\Sigma_2}$.
If a multitrace is not locally correct on location 2,
$\decentral{\aut,\Sigma_2}$ will not accept it, e.g.
$\mu=(e,cd,bbbb)$.
\end{example} 

\begin{theorem}
\label{thm:decentral}
Let $\aut$ an NFA over alphabet 
$\Sigma = \Sigma_i \uplus \Sigma_\noti$.
Then $\decentral{\aut,\Sigma_i}$ recognizes
 $\lang(\aut)$ projected onto $\Sigma_i$, i.e.
$\lang(\decentral{\aut,\Sigma_i})=\lang(\aut)_\proj{i}.$
\end{theorem}

\begin{proof}
We replace the transitions of $\aut$
over letters of $\Sigma_\noti$ by 
$\varepsilon$-transitions, and then
follow the classic conversion from 
$\varepsilon$-NFA to DFA
to obtain $\decentral{\aut,\Sigma_i}$.
This conversion guarantees that the DFA
recognizes the same language as the $\varepsilon$-NFA,
which is easily seen to accept exactly 
the words of $\lang(\aut)_\proj{i}$.
\end{proof}

Given an NFA
$\aut = (Q, \Sigma, \delta, q_0, F)$,
let $\Delta_\noti(q)=\delta \cap \set{q'\reads{b}q'' \ | \ q',q'' \in \delta_\noti(q), b\in \Sigma_\noti}$ 
be the set of transitions of $\aut$
corresponding to $\delta_\noti(q)$, and let
$\Delta(q,a)=\set{q\reads{a}q' \ | \ q' \in \delta(q,a)}$
be the set of transitions of $\aut$
corresponding to $\delta(q,a)$.
We extend these definitions naturally to subsets $S$ of $Q$, and use them to define our labeling function.

Each $a$-transition $t$ of $\decentral{\aut,\Sigma_i}$ 
is labeled by the sub-automaton
of $\aut$ containing the corresponding $a$-transition 
as well as the $\Sigma_\noti$-transitions that can be
taken after it.
We define our labeling function.

\begin{definition}
\label{def:ar}
Let $\aut = (Q,\Sigma,\delta,q_0,F)$ 
an NFA over $\Sigma = \Sigma_i\uplus\Sigma_\noti$, 
and let $t = S \reads{a} S'$ a transition of $\decentral{\aut,\Sigma_i}$.
Then $\ar{t}$ is the sub-automaton of $\aut$
defined by the set of transitions
$\Delta(S,a) \cup \Delta_\noti(\delta(S,a)).$
\end{definition}

\begin{example}
\cref{fig:areas} shows the $\ar{t}$
for transitions $t$ of $\decentral{\aut,\Sigma_2}$
(\cref{ex:proj}).
From left to right:  
for $t=\set{0,1,3}\reads{c}\set{2,3}$,
for $t=\set{0,1,3}\reads{d}\set{4,1}$,
and the third NFA is for both
$t=\set{4,1}\reads{c}\set{2,3}$
and for $t=\set{4,1}\reads{c}\set{4,1}$.
\end{example}

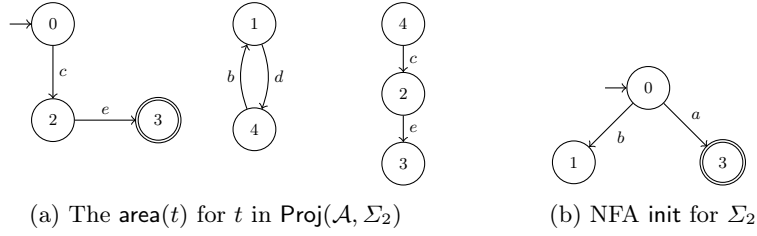
\begin{figure}[t]
	\centering
	\begin{subfigure}[t]{.55\textwidth}
	\centering
    \begin{tikzpicture}[node distance=2cm, auto, initial text={}, scale=0.7, transform shape]
		\node[state, initial] (q0) {0};
		\node[state, below=1 of q0] (q2) {2};
		\node[state, accepting, right of=q2] (q3) {3};
		
		\node[state, right=3 of q0] (q1) {1};
		\node[state, below of=q1] (q4) {4};
		
		\node[state, right=2 of q1] (q44) {4};
		\node[state, below=0.5 of q44] (q22) {2};
		\node[state, below=0.5 of q22] (q33) {3};

		\path[->]
		(q0) edge[] node {$c$} (q2)
		(q2) edge[] node {$e$} (q3)
		(q4) edge[bend left=20] node {\(b\)} (q1)
		(q1) edge[bend left=20] node {\(d\)} (q4)
		(q44) edge[] node {$c$} (q22)
		(q22) edge[] node {$e$} (q33)
		;
\end{tikzpicture}
    \caption{The $\ar{t}$ for $t$ in $\decentral{\aut,\Sigma_2}$}
    \label{fig:areas}
    \end{subfigure}
~
    \begin{subfigure}[t]{.35\textwidth}
	\centering
    \begin{tikzpicture}[node distance=2cm, auto, initial text={}, scale=0.7, transform shape]
		\node[state, initial] (q0) {0};
		\node[state, below left of=q0] (q1) {1};
		\node[state, accepting, below right of=q0] (q2) {3};

		\path[->]
		(q0) edge[] node {$b$} (q1)
		(q0) edge[] node {$a$} (q2)
		;
\end{tikzpicture}
	\caption{NFA $\init$ for $\Sigma_2$}
	\label{fig:init}
	\end{subfigure}
	\caption{Illustration of \cref{def:ar} and \cref{def:init}.}
	
\end{figure}

We need to define one final object 
to complete our local verification set-up.
Intuitively, since the labels of the transitions
of $\decentral{\aut,\Sigma_i}$ encode an area of $\aut$
corresponding to a sequence of the form 
$\reads{a} \reads{w}$ for $w\in\Sigma_\noti^*$,
we need to cover the case where there exists an 
initial sequence $\reads{w}$ for $w\in\Sigma_\noti^*$
from $q_0 $.
One can think of this as a special label on
the initial state $q_0^P$ of $\decentral{\aut,\Sigma_i}$.

\begin{definition}
\label{def:init}
Let $\aut = (Q,\Sigma,\delta,\set{q_0},F)$ 
an NFA over $\Sigma = \Sigma_i\uplus\Sigma_\noti$.
Then $\init$ is the sub-automaton of $\aut$
defined by the set of transitions
$\Delta_\noti(q_0 ).$
\end{definition}

\begin{example}
\cref{fig:init} shows $\init$ computed for $\aut$ of \cref{ex:proj} and $\Sigma_2=\set{c,d}$.
\end{example}

\subsection{First step: Local Check}
\label{sec:local-check}

In this section, we show how $\decentral{\aut,\Sigma_i}$, 
labeling function $\textsf{area}$ and $\init$
are used to perform a local check of a local trace.
Given a local trace $\mu_i \in \Sigma_i^*$,
we define the output $\area{\mu_i}$ obtained by 
reading $\mu_i$ on $\decentral{\aut,\Sigma_i}$.

\begin{definition}
\label{def:area}
Let $\aut$
an NFA over $\Sigma = \Sigma_i\uplus\Sigma_\noti$, 
and let 
$\decentral{\aut,\Sigma_i} = (Q^P,\Sigma_i,\allowbreak \delta^P,\set{q_0^P},F^P)$.
Let $\mu_i\in \Sigma_i^*$ a word, 
and let $t_1, \ldots, t_k$ with $t_j\in\delta^P$ 
the unique path reading $\mu_i$ 
from $q_0^P$ in $\decentral{\aut,\Sigma_i}$.
Then $\area{\mu_i}$ is the NFA defined as the union 
$\init \cup \bigcup_{j=1}^k \ar{t_j}.$
That is, $\area{\mu_i}$ is the NFA over $\Sigma$
of state set $Q^\init\cup\cup_j Q^{\ar{t_j}}$,
transitions $\delta^\init\cup\cup_j \delta^{\ar{t_j}}$,
initial state $q_0^\aut$ and final states
$F^\init\cup\cup_j F^{\ar{t_j}}$.
\end{definition}

\begin{figure}[t]
	\centering
	\begin{subfigure}[t]{.45\textwidth}
	\centering
    \begin{tikzpicture}[node distance=2cm, auto, initial text={}, scale=0.7, transform shape]
		\node[state, initial] (q0) {0};
		\node[state, below left of=q0] (q1) {1};
		\node[state, below=2 of q0] (q2) {2};
		\node[state, accepting, right of=q2] (q3) {3};

		\path[->]
		(q0) edge[] node {$b$} (q1)
		(q0) edge[] node {$c$} (q2)
		(q0) edge[] node {$a$} (q3)
		;
\end{tikzpicture}
    \caption{The NFA $\area{c}$}
    \label{fig:nfa-area-1}
    \end{subfigure}
~
    \begin{subfigure}[t]{.45\textwidth}
	\centering
    \begin{tikzpicture}[node distance=2cm, auto, initial text={}, scale=0.7, transform shape]
		\node[state, initial] (q0) {0};
		\node[state, below left of=q0] (q1) {1};
		\node[state, below=2 of q0] (q2) {2};
		\node[state, accepting, right of=q2] (q3) {3};
		\node[state, left of=q2] (q4) {4};

		\path[->]
		(q0) edge[] node {$b$} (q1)
		(q0) edge[] node {$a$} (q3)
		(q2) edge[] node {$e$} (q3)
		(q4) edge[bend left=20] node {\(b\)} (q1)
		(q1) edge[bend left=20] node {\(d\)} (q4)
		(q4) edge[] node {$c$} (q2)
		;
\end{tikzpicture}
    \caption{The NFA $\area{dddc}$}
    \label{fig:nfa-area-2}
	\end{subfigure}
	\caption{Illustration of \cref{def:area}.}

	\label{fig:nfa-area}
\end{figure}

\begin{example}
\cref{fig:nfa-area} shows $\area{c}$ and $\area{dddc}$ 
for $\aut$ of \cref{ex:proj} and $\Sigma_2=\set{c,d}$.
Intuitively, $\area{\mu}$ contains the area of $\aut$
which was traversed 
when reading an accepting word 
that projects onto $\mu$.
If the first letter of $\mu$ is a $c$,
any word of $\lang(\aut)$ that projects onto $\mu$
must take transition $0 \reads{c} 2$.
On the contrary, if the first letter is a $d$, 
transition $0 \reads{c} 2$ cannot be taken.
\end{example}

The initial state of $\area{\mu_i}$ 
is $q_0$ of $\aut$
because $\init$ contains $q_0$,
and $\init$ and the $\ar{t}$ are sub-automata of $\aut$.
The constructed $\area{\mu_i}$,
which is also a sub-automaton of $\aut$,
accepts the words of $\lang(\aut)$ 
that are equal to $\mu_i$
when projected onto letters of $\Sigma_i$.

\begin{theorem}
\label{thm:area}
Let $\aut$ an NFA over alphabet 
$\Sigma = \Sigma_i \uplus \Sigma_\noti$,
and $\mu_i$ a word of $\Sigma_i^*$.
Then 
$\set{w \in \lang(\aut) \ | \ w_\proj{i} = \mu_i} \subseteq \lang(\area{\mu_i}).$
\end{theorem}

Notice that $\area{\mu}$ may accept more
than the words of $\aut$ that project onto $\mu$.
For example $\area{dddc}$ of \cref{fig:nfa-area-2}
accepts $bdce$.

\begin{proof}
Let $q_0$ the initial state of $\aut$,
and $Q$ its set of states.
Let $\decentral{\aut,\Sigma_i}$ be the DFA
defined in \cref{def:decentral-projection},
and let $q_0^P$ its initial state
and $Q^P$ its set of states.
Let $w\in\lang(\aut)$ such that 
$w_\proj{i} = \mu_i$.
If $w$ is the empty word,
then $q_0$ is final in $\aut$.
Thus $q_0$ is also final in $\init$
and also in $\area{\mu_i}$ by definition;
thus $\varepsilon\in\lang(\area{\mu_i}))$.

Suppose $w$ is not the empty word.
Word $w$ can be decomposed into
$w = w_0 a_1 w_1 a_2 \ldots a_k w_k$
such that $a_j\in\Sigma_i$ and 
$w_j \in \Sigma_\noti^*$ for all $j$,
and $a_1 a_2 \ldots a_k = \mu_i$.
There exists an accepting path in $\aut$ 
$q_0\reads{w_0} q_1\reads{a_1} p_1\reads{w_1} q_2\reads{a_2} p_2\ldots q_k\reads{a_k} p_k\reads{w_k} q_{k+1}$
We will show that this accepting path 
is in $\area{\mu_i}$ too.
Let $S_0:=q_0^P \reads{a_1} S_1 \reads{a_2} S_2 \ldots \reads{a_k} S_k$
be the unique path for $\mu_i$ 
in $\decentral{\aut,\Sigma_i}$. 

By \cref{def:init}, 
$q_0\reads{w_0} q_1$ is in $\init$,
and thus also in $\area{\mu_i}$.
Let us show by induction that 
for all $j\in [1,k]$,
$q_j \in S_{j-1}$ and 
$q_j\reads{a_j} p_j\reads{w_j} q_{j+1}$ is in 
$\ar{S_{j-1}\reads{a_j} S_j}$, and thus 
in $\area{\mu_i}$.
For $j=1$, notice that since 
$q_0^P = \delta_\noti(q_0)$,
$q_1$ is in $S_0=q_0^P$.
Transition $q_1\reads{a_1}p_1$
is in $\Delta(S_0,a_1)$
and $p_1\reads{w_1} q_2$ is in 
$\Delta_\noti(\delta(S_0,a_1))$,
and so $q_1\reads{a_1} p_1\reads{w_1} q_2$
is in $\ar{S_0\reads{a_j} S_1}$ 
(\cref{def:ar}).

Suppose that the induction holds 
for some $j-1 \ge 1$.
By induction $q_j \in S_{j-1}$.
As above, transition  $q_j\reads{a_j}p_j$
is in $\Delta(S_{j-1},a_j)$
and $p_j\reads{w_j} q_{j+1}$ is in 
$\Delta_\noti(\delta(S_{j-1},a_j))$.
Thus $q_j\reads{a_j} p_j\reads{w_j} q_{j+1}$ 
is in $\ar{S_{j-1}\reads{a_j} S_j}$
and we are done.
\end{proof}

The first step of the procedure is local:
each local verifier reads $\mu_i$ on $\decentral{\aut,\Sigma_i}$:
if the unique path is not accepting, 
the local verifier sends message \localerror{i}
to the central verifier and the procedure stops.
Otherwise, it proceeds to the second step.

\subsection{Second Step: Central Check}
\label{sec:central-check}

In this section, we show how the 
$\area{\mu_i}$
can be used
to determine if there exists an interleaving 
of the $\mu_i$ in $\lang(\aut)$.
We define the intersection automaton $\interaut$
of the $\area{\mu_i}$.
We will show that if a candidate interleaving exists, 
it can be read on $\interaut$.

\begin{definition}
\label{def:inter}
For each $i\in [1,n]$,
let $\mu_i$ a word of $\Sigma_i^*$, and
let $\area{\mu_i} = (Q^i,\Sigma, \delta^i,\set{q_0},F^i)$ 
the NFA of Definition \ref{def:area} for $\aut$ and $\Sigma_i$.
Then $\interaut$ is the NFA defined as the intersection
$\bigcap_{i=1}^n \area{\mu_i}.$
That is, $\interaut$ is the NFA over $\Sigma$
of state set $\cap_i Q^i$, 
transitions $\cap_i \delta^i$,
initial state $q_0$ and final states $\cap_i F^i$.
\end{definition}

\begin{figure}[t]
    \centering
    \begin{tikzpicture}[node distance=2cm, auto, initial text={}, scale=0.7, transform shape]
		\node[state, initial] (q0) {0};
		\node[state, below left of=q0] (q1) {1};
		\node[state, below=2 of q0] (q2) {2};
		\node[state, accepting, right of=q2] (q3) {3};
		\node[state, left of=q2] (q4) {4};

		\path[->]
		(q0) edge[] node {$b$} (q1)
		(q2) edge[] node {$e$} (q3)
		(q4) edge[bend left=20] node {\(b\)} (q1)
		(q1) edge[bend left=20] node {\(d\)} (q4)
		(q4) edge[] node {$c$} (q2)
		;
\end{tikzpicture}
    \caption{NFA $\textsf{Inter}(\aut, (e,dddc,bb))$}
    \label{fig:nfa-inter}
\end{figure}
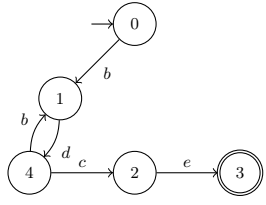

It follows from construction that $\interaut$ recognizes
exactly the words $w$ such that 
$w\in\lang(\area{\mu_i})$ for all $i$.
Combining \cref{thm:area} and that fact that
the $\area{\mu_i}$ and $\interaut$
are sub-automata of $\aut$ 
we obtain the following result. 
We shorten $\interaut$ to $\inter$ when it is clear from context.

\begin{theorem}
\label{thm:inter}
Let $\aut$ an NFA over 
$\Sigma = \uplus_{i=1}^n \Sigma_i$, 
and let $\mu_i$ a word of $\Sigma_i^*$
for each $i\in [1,n]$.
Then
$\shuffle_{i=1}^n \ \mu_i \cap \lang(\aut) \subseteq \lang(\interaut) \subseteq \lang(\aut).$
\end{theorem}

\begin{proof}
Let $w$ in $\shuffle_{i=1}^n \ \mu_i \cap \lang(\aut)$.
We have $w$ in $\set{w \in \lang(\aut) \ | \ w_\proj{i} = \mu_i}$ for all $i$.
By \cref{thm:area}, $w$ is in 
$\lang(\area{\mu_i})$ for all $i$.
It follows from the definition that 
$\lang(\inter) = \cap_{i=1}^n \lang(\area{\mu_i})$,
and thus we have the left inclusion.
The right inclusion follows from the fact
that the $\area{\mu_i}$
are sub-automata of $\aut$.
\end{proof}

The second step of the semi-centralized procedure is to apply the centralized procedure (Algorithm \ref{alg:central})
on $\interaut$ and $(\mu_i)_i$.
Intuitively, either a local incompatibility with $\aut$
was detected in the first step, 
or the incompatibility lies in
the interleaving of the $\mu_i$.
The second step will detect this using the information
generated in the first step: transitions of $\aut$
that are not in an area visited by a $\mu_i$
are discarded, and the centralized procedure
is applied to this smaller automaton.

\begin{example}
Figure \ref{fig:nfa-inter} shows  $\textsf{Inter}(\aut, (e,dddc,bb))$ for $\aut$ of \cref{ex:proj}. 
It is $\aut$ from which the transitions we know have not been visited by the multitrace are removed.
Multitrace $(e,dddc,bb)$ is locally correct, 
each local trace at $i$ is accepted by $\decentral{\aut,\Sigma_1}$.
However it is not accepted by $\aut$,
and applying Algorithm \ref{alg:central} on $\textsf{Inter}(\aut, (e,dddc,bb))$ will return \fail.
\end{example}

\begin{corollary}
\label{coro:inter}
Let $\aut$ an NFA over 
$\Sigma = \uplus_{i=1}^n \Sigma_i$, 
and let $\mu_i$ a word of $\Sigma_i^*$
for each $i\in\set{1,\ldots,n}$.
Then $\central{\interaut,(\mu_i)_i}$ recognizes
the language $\shuffle_{i=1}^n \ \mu_i \cap \lang(\aut)$.
\end{corollary}

\begin{proof}
By \cref{thm:central}, $\central{\inter,(\mu_i)_i}$ 
recognizes the language 
$\allowbreak \lang^\textsf{Central} \allowbreak := \shuffle_{i=1}^n \ \mu_i \cap \lang(\inter)$.
By \cref{thm:inter}, 
$\lang(\inter) \subseteq \lang(\aut)$ and thus 
$\lang^\textsf{Central} \subseteq \shuffle_{i=1}^n \ \mu_i \cap \lang(\aut).$
We get the other direction using the fact that 
$\shuffle_{i=1}^n \ \mu_i \cap \lang(\aut)
\subseteq \shuffle_{i=1}^n \ \mu_i$
and that $\shuffle_{i=1}^n \ \mu_i \cap \lang(\aut)
\subseteq \lang(\inter)$ (\cref{thm:inter}).
\end{proof}

If no candidate interleaving exists, 
there are many cases in which this will 
be apparent on the intersection automaton,
which will have no accepting path.
To detect this before applying Algorithm \ref{alg:central}, we introduce a preprocessing step
in which we check that $\interaut$ has at least one
accepting path.

\begin{example}
Consider multitrace $(e,dddc,\varepsilon)$.
$\area{\varepsilon}$ for $\aut$ and $\Sigma_3=\set{b}$ is equal to the NFA of three transitions $0 \reads{c} 2, 0 \reads{a} 3$ and $2 \reads{e} 3$. 
$\textsf{Inter}(\aut, \allowbreak(e,dddc,\varepsilon))$ is equal to  state $q_0$ with no transitions, and no accepting paths.
\end{example}

\subsection{Procedure}
\label{sec:decentral-correct}

The local verifier set-up consists in constructing,
for each location $i$,
DFA $\decentral{\aut, \allowbreak \Sigma_i}$, 
its transitions labels $\area{t}$ for each transition $t$,
as well as $\init$.
Given an NFA $\aut$ over 
$\Sigma = \uplus_{i=1}^n \Sigma_i$
and a family of local traces $(\mu_i)_{i=1}^n$ 
such that $\mu_i \in \Sigma_i^*$,
the semi-centralized procedure works as follows.
\begin{enumerate}
    \item \emph{Local checks:} for each location $i$,    
the local verifier for $i$ reads $\mu_i$
on $\allowbreak \decentral{\aut,\Sigma_i}$.
The local verifier sends \localerror{i}
to the central verifier if $\mu_i$ is not accepted
and the procedure stops.
Otherwise, it sends $\area{\mu_i}$
to the central verifier.
This is illustrated in Algorithm \ref{alg:local-check}.
    \item \emph{Central check:} the central verifier constructs $\interaut$.
If $\inter$ has no accepting path,
the central verifier sends \intererror{} 
and the procedure stops.
Otherwise, the central verifier applies the
centralized procedure of Algorithm \ref{alg:central}
on the $\mu_i$ and $\inter$: 
it constructs $\central{\inter,(\mu_i)_i}$
and answers \pass{} if there is a final state;
otherwise it answers \centralerror{}.
This is illustrated in Algorithm \ref{alg:central-check}.
\end{enumerate}
\localerror{i} gives the information that 
the multitrace is not locally correct because of location $i$.
\intererror{} gives the information that the multitrace is locally correct, but that there is no common path for accepting words projecting onto the local traces.
Finally, \centralerror{} gives the information that the multitrace is locally correct, but no interleaving of the local traces is accepted.

\begin{algorithm}[t]
\caption{Local check for location $i$}
\label{alg:local-check}
\Input{DFA $\decentral{\aut,\Sigma_i} = (Q,\Sigma_i,\delta,\set{q_0},F)$, $\ar{t}$ for all $t \in \delta$, $\init$ and trace $\mu_i \in \Sigma_i^*$}
\Output{\localerror{i} xor NFA $\area{\mu_i}$}
$\area{\mu_i} \gets \init$\;
$S \gets q_0$\; 
\textbf{if} $q_0 \in F$ \textbf{then} \KwRet{} \pass{}\;
\For{$a$ \textbf{in} $\mu_i$}{
    $S \gets \delta(S)$\;
    \textbf{if} $S=\emptyset$ \textbf{then} \KwRet{} \localerror{i}\;
    \textbf{if} $S \in F$ \textbf{then} \KwRet{} \pass{}\;
}
\KwRet{} \localerror{i}
\end{algorithm}

\begin{algorithm}[t]
\caption{Central check}
\label{alg:central-check}
\Input{pairs $(\mu_i, \area{\mu_i})$ of trace and NFA for each location $i$}
\Output{\pass{} or \intererror{} or \centralerror{}}
$\interaut \gets \cap_{i=1}^n \area{\mu_i}$ \;
\uIf{$\interaut$ has no accepting path}{\KwRet{} \intererror{}}
\Else{apply \textsf{Algorithm1}$(\inter,(\mu_i)_i)$ 

}
\end{algorithm}

\paragraph{Soundness and Completeness}
The semi-centralized procedure is \emph{sound}: 
if an \fail{} message is sent,
then there exists no interleaving of the $\mu_i$
that is recognized by $\aut$.
There are three cases in which a \fail{} message is sent.
\begin{itemize}
    \item 
Case 1: \localerror{i} is sent by a local verifier
at location $i$ in the first step.
This implies that $\mu_i$ is not
in $\lang(\aut)_\proj{i}$ (\cref{thm:decentral}).
Suppose there exists $w$, an interleaving of the $\mu_i$
that is recognized by $\aut$, i.e.
$w\in \shuffle_{i=1}^n \ \mu_i \cap \lang(\aut)$.
Then $w_\proj{i}=\mu_i$ and 
$w_\proj{i}\in \lang(\aut)_\proj{i}$, a contradiction.
    \item
Case 2: \intererror{} is sent by the central verifier
in the second step.
This implies $\interaut$ recognizes the empty language,
and thus
$\shuffle_{i=1}^n \ \mu_i \cap \lang(\aut)$
is empty, by \cref{thm:area}.
    \item
Case 3: \centralerror{} is sent by the central verifier.
Thus $\interaut$ recognizes the empty language, 
and this implies that 
$\shuffle_{i=1}^n \ \mu_i \cap \lang(\aut)$
is empty, by \cref{coro:inter}.
\end{itemize}
The semi-centralized procedure is \emph{complete}: 
if there exists no interleaving of the $\mu_i$
that is recognized by $\aut$,
then a \fail{} message is sent.
We reason by contraposition:
suppose the central verifier sends \pass.
This means the construction of 
$\allowbreak \central{\inter,(\mu_i)_i}$
marks a state final.
Since only states reachable 
from the initial state are built, 
this implies the existence of an accepting path.
 Thus there exists
$w\in\shuffle_{i=1}^n \ \mu_i \cap \lang(\aut)$
(\cref{coro:inter}).

\paragraph{Complexity}
We consider two complexities:
the complexity of setting up the procedure 
given $\aut$, 
and the complexity of the two-step procedure
which checks a given family of local traces
$(\mu_i)_{i=1}^n$ against $\aut$.

 \textit{Set-up:} \
Given $\aut$ and $\Sigma_i$, 
the set-up consists in building 
DFA $\decentral{\aut,\Sigma_i}$, 
function $\textsf{area}$ and $\init$.
$\decentral{\aut,\Sigma_i}$
has at most $2^{\size{Q^\aut}}$ states.
Since it is a DFA, it has at most
$2^{\size{Q^\aut}} \size{\Sigma_i}$ transitions.
Each $\ar{t}$ as well as $\init$ 
is a subautomaton of $\aut$,
i.e. an NFA with at most $\size{Q^\aut}$ states.
Therefore, the time complexity 
of the procedure set-up 
is exponential in $\aut,\Sigma_i$.
However, this is a worst-case
complexity: in practice we only build 
the subset-states of $\decentral{\aut,\Sigma_i}$
reachable from its initial state, 
and usually far from all elements of $2^{Q^\aut}$ 
are states of $\decentral{\aut,\Sigma_i}$.
This can be seen on the experiments in Section \ref{sec:implementation}.

 \textit{Procedure:} \
Given $(\mu_i)_i$,
the first step consists in each local verifier
reading $\mu_i$ on $\decentral{\aut,\Sigma_i}$
and outputting $\area{\mu_i}$:
this takes linear complexity in $\mu_i$.
DFAs have unique paths on words and $\area{\mu_i}$
is simply the union of the labels of the path.
If $\mu_i$ is not accepted (\localerror{i}), the procedure stops here.
The second step builds $\interaut$, 
sub-automaton of $\aut$,
to check for an accepting path: 
this takes complexity $n \size{\aut}$,
for the $n$ NFAs $\area{\mu_i}$.
If \intererror, the procedure stops here.
Otherwise the procedure applies the centralized
procedure to $(\mu_i)_i$ and $\interaut$,
which takes quadratic complexity 
as seen in \cref{sec:proc-central}.

\paragraph{Evaluation}
This approach is reusable:
several traces can be checked against one interaction model
without having to rebuild the local verifiers.
The initial cost of computing the $\decentral{\aut,\Sigma_i}$
is mitigated by the fact that they are reused 
for many local traces.
The two-step procedure is then essentially a 
pre-processing which finds local errors fast if they exist,
and which otherwise applies the centralized procedure on a smaller NFA.
The procedure is especially efficient in the case of local errors.
The idea is that the local verifiers can do their work in parallel, sending the answers to the central verifier.
Local errors are detected at very little complexity-cost,
and do not require communication between local verifier and central verifier (which may be costly).

\section{Implementation}
\label{sec:implementation}

We implemented our two procedures in a prototype tool \cite{prototype-rvia-anonymous}.
Our tool takes an NFA and a set of multitraces as input.
The NFA must have $n$ subsystems defined as disjoint sub-alphabets of the NFA's alphabet (these are added by the user or randomly generated),
and the multitraces should contain $n$ local traces that are words of the corresponding sub-alphabet.
The tool has two commands:
\textbf{AC} applies the centralized procedure and returns \pass{} or \fail;
    \textbf{AS} applies the semi-centralized procedure and returns \pass{}, \intererror, \centralerror{} or  \localerror{i} for $i$ a subsystem name.
We coded in the Rust language and straightforwardly followed the pseudo-code algorithms of this paper.
In our experiments we make the simplifying assumption that each location has exactly one subsystem.
The experiments were run on a computer with an
Intel(R) Core(TM) i7-13800H (2.50 GHz).
The results are summarized in Table \ref{tab:grouped_experiments}.

\sisetup{
  round-mode = places,
  round-precision = 3,
  scientific-notation = true,
  output-exponent-marker = \text{e}  
}
\begin{table}[t]
\centering
\caption{Experimental results grouped by NFA, with different error scenarios}
\label{tab:grouped_experiments}
\renewcommand{\arraystretch}{1.2}
\setlength{\tabcolsep}{5pt}

\begin{tabular}{l l l c c c}
\toprule
\textbf{NFA} & \textbf{Verdict} & {\textbf{Length / \#MT}} &
{\textbf{AC(s)}} & {\textbf{AS(s)}} & {\textbf{Setup(s)}} \\
\midrule

\multirow{2}{*}{\textbf{NFA1}} 
 & Pass         & [200 -250]/1000  & \num{0.1255524} & \num{9.323158}  & \num{0.001763}  \\
 & LocalError     & [200-250]/196  & \num{ 0.00071156} & \num{ 0.00140382} &\num[scientific-notation = true, round-precision=3]{0.00133524}  \\
 states: 94 
 & InterError    &[200-250]/2004  &  \num{0.00852422}  &  \num{0.13879746}  &  \num[scientific-notation = true, round-precision=3]{0.00167436}  \\
 trans.: 320
 & CentralError   & [200-250]/910 &   \num{1.33229162}  &  \num{50.15003662}  &  \num[scientific-notation = true, round-precision=3]{0.0013959}  \\
 & AllTypes     & [200-250]/1000 & \num{0.1127796} & \num{7.99077106}  &\num[scientific-notation = true, round-precision=3]{0.00163412}  \\
\midrule

\multirow{2}{*}{\textbf{NFA2}} 
 & Pass         & [200-250]/1000 & \num{0.0991647} & \num{9.3375662} & \num[scientific-notation = true, round-precision=3]{0.0002053} \\
 & LocalError     & [200-250]/1339  &\num{4.06846654} & \num{0.16494372} & \num[scientific-notation = true, round-precision=3]{0.00019096}\\
 states: 35
 & InterError      & [200-250]/1569 & \num{0.00463964}&\num{0.03992106}  & \num[scientific-notation = true, round-precision=3]{0.00018442}\\
 trans.: 23
 & CentralError   & [200-250]/358 &\num{0.22848372} & \num{9.56081716} &\num[scientific-notation = true, round-precision=3]{0.00019376} \\
 & AllTypes      & [200-250]/1000  &\num{1.40773182} & \num{4.801491} & \num[scientific-notation = true, round-precision=3]{0.00018528}\\

\midrule

\multirow{2}{*}{\textbf{NFA3}} 
 & Pass         & [200-250]/1000 & \num{0.0991545} & \num{9.1013433} & \num{9.18e-5} \\
 & LocalError    & [200-250]/2017  &\num{0.16728364} & \num{0.00017134} &\num{5.392e-5}\\
 states: 15
 & InterError  & [200-250]/1569 &\num{0.0117915} & \num{0.01007384} &\num{6.354e-5}\\
 trans.: 75
 & CentralError  & [200-250]/1015 &\num{0.00620678} & \num{25.0345444} &\num{5.658e-5}\\
 & AllTypes    & [200-250]/1000 &\num{0.08015528} & \num{2.69461402 } &\num{5.97e-5}\\

\midrule

\multirow{2}{*}{\textbf{NFA4}} 
 & Pass          & [200-250]/1000 & \num{0.5123971} & \num{9.9814967} & \num{0.0014573} \\
 & LocalError    & [200-250]/1971  & \num{0.6600194} & \num{0.31546424} & \num[scientific-notation = true, round-precision=3]{0.00146932} \\
 states: 124
 & InterError      & [200-250]/949 & \num{0.14561072} & \num{0.0489298} & \num[scientific-notation = true, round-precision=3]{0.00124302} \\
 trans.: 349
 & CentralError    & [200-250]/1114 & \num{0.2744528} & \num{29.05818554} & \num[scientific-notation = true, round-precision=3]{0.00160478} \\
 & AllTypes       & [200-250]/1000  & \num{0.39412858} & \num{3.02697582} & \num[scientific-notation = true, round-precision=3]{0.00143856} \\

\bottomrule
\end{tabular}
\end{table}

The NFAs 
were selected from the AutomatArk Benchmark \cite{automatark},
which itself collects automata from different benchmark sources.
We selected NFAs that are varied in source, size and in shape; 
we restrict ourselves to four in this paper for length considerations.
For each NFA, we call \textbf{AC} and \textbf{AS} on 
sets of multitraces that all return \textsf{verdict}, for 
\textsf{verdict} equal to \pass{}, \intererror, \centralerror{} or  \localerror{}.
We additionally run the commands on a set of multitraces with mixed return verdicts.
The sets of multitraces are generated using randomness:
for \pass{} we traverse the NFA to generate correct traces which we then project into a multitrace;
for \localerror{}, we take a correct multitrace and introduce errors;
for \intererror{} we mix local traces from correct multitraces;
finally, for \centralerror{} we take a correct trace with cycles and iterate the cycle a different number of times on each subsystem.

Given a set of multitraces for a given verdict,
the number of multitraces is \textbf{\#MT}, 
the length of each multitrace is in the interval \textbf{Length}, and
the average time for running a command on the set is given in seconds. 
When the command is \textbf{AS}, the average set-up time (constructing the $\decentral{\aut,\Sigma_i}$) is also given to show that in practice in takes very little.
Both procedures are very efficient, and \textbf{AS} is more efficient on verdict \localerror{} except in one case where both times are extremely small.

\section{Discussion}
\label{sec:related}
\paragraph{Related Work}
This paper studies RV for interactions and multitraces,
on systems which are distributed and with no global clock.
Interactions were introduced in \cite{MaheGG20} 
and have been studied as specifications for RV approaches
that consider traces, multitraces or incomplete traces (see Mahe's doctoral thesis \cite{PhdMahe21}).
For a full positioning of interactions with regards to other specification languages, we send the reader to Section 8 of \cite{MaheBGG25}.
This paper has similarities with work on \textit{multiparty session types} (MST).
In MST literature, the goal is to synthesize local specifications  from a  global specification such that if subsystems satisfy the local specifications 
then the whole system satisfies the global specification and is deadlock free. 
Global specifications for which this is possible are called \textit{implementable}.
In \cite{LiSWZ23}, the authors model the global specification as an NFA, and synthesize local specifications by projecting the NFA onto the subsystems (like our  \textsf{Proj}) when the global is implementable.
We study global specifications even when they are not implementable in the MST sense, by allowing subsystems to communicate.
Some papers use NFA as (global) specifications for \textit{decentralized} RV \cite{ElHokayemF20,FalconeCF14}.
They are similar to our approach in that local verifiers read local traces and communicate to build a global verdict.
However, they assume a global clock, and all subsystem emit actions synchronously.
A paper close to ours is \cite{MostafaB15}: it considers a distributed system with local verifiers and no global clock. The specifications are in LTL$_3$ \cite{BauerLS11}, an adaptation of LTL which outputs (3) verdicts on finite traces. Each subsystem has a local verifier which has a copy of the (deterministic) specification automaton, and the local verifiers communicate regularly using vector clocks.
The verifiers keep a set of possible verdicts, 
and the algorithm guarantees that if a correct interleaving exists, one of the verifiers has the correct verdict (similar to our condition of one correct interleaving being enough).
This paper uses a different technique from ours, relying on more communication between verifiers. Their problem is almost the same as ours, except that verdicts (instead of \pass/\fail) are output, and multitraces are seen as finite prefixes of infinite multitraces.

\paragraph{Future Work}
Our work is cast in the light of interactions, 
but can be applied as soon as the specification
is given as an NFA and if the subsystem alphabets are disjoint.
In future, we want to exploit the fact that actions are either sends or receives to deliver a finer analysis.
Another avenue for improvement is to consider finite prefixes of multitraces, inspired by the verdicts of \cite{MostafaB15},
as well as multitraces in which message loss may occur.
Communication protocols are often designed in a parameterized manner, e.g. one server and $n$ clients.
Another research avenue is to adapt interactions to express this parameterization.

\subsubsection{Acknowledgements}
This work was funded by the French government under the France 2030 ANR program “PEPR Networks of the Future” (ref. ANR-22-PEFT-0009).

\bibliographystyle{plain}
\bibliography{ref}

\end{document}